\pdfoutput=1
\documentclass[journal]{IEEEtran}

\usepackage{graphicx}
\usepackage{xcolor}
\usepackage{url}
\usepackage{tikz}
\usepackage{framed}
\usepackage{cite}
\usepackage{ifthen}
\usepackage[hidelinks]{hyperref}

\usepackage{amsmath}
\usepackage{amssymb}
\usepackage{amsfonts}
\usepackage{bm}
\usepackage{amsthm}
\usepackage{mathtools}
\usepackage{latexsym}
\usepackage{empheq}

\usepackage{algorithm}
\usepackage{algorithmic}
\usepackage{booktabs}
\usepackage{multirow}
\usepackage{enumerate}
\usepackage{siunitx}
\usepackage[caption=false,font=footnotesize]{subfig}
\usepackage{array}

\newtheorem{theorem}{Theorem}
\newtheorem{remark}{Remark}
\newtheorem{prop}{Proposition}

\newtheorem{problem}{Problem}

\newcommand{\R}{\mathbb{R}}
\newcommand{\C}{\mathbb{C}}

\DeclareMathOperator{\tr}{tr}

\newcommand{\INLINEIF}[2]{\ifthenelse{#1}{#2}{}}

\mathtoolsset{showonlyrefs=true}  

\newcommand{\tblack}{\textcolor{black}}

\DeclareUnicodeCharacter{221E}{\ensuremath{\infty}}

\begin{document}
\title{\tblack{Data-Driven Regularized Time-Limited $h^2$ Model Reduction from Noisy Impulse Responses}}
\author{Hiroki Sakamoto and Kazuhiro Sato\thanks{H. Sakamoto and K. Sato are with the Department of Mathematical Informatics, Graduate School of Information Science and Technology, The University of Tokyo, Tokyo 113-8656, Japan, email: shiroki2875@gmail.com (H. Sakamoto), kazuhiro@mist.i.u-tokyo.ac.jp (K. Sato) }}

\maketitle
\thispagestyle{empty}
\pagestyle{empty}


\begin{abstract}
This paper develops a data-driven time-limited $h^2$ model reduction method for discrete-time linear time-invariant systems. Specifically, we formulate and solve a \tblack{regularized} time-limited $h^2$ model reduction problem using only \tblack{noisy} impulse response data.
Furthermore, we show that the objective function and its gradient can be represented using only \tblack{noisy} impulse response data.
Numerical experiments using SLICOT benchmarks demonstrate that the proposed regularized method 
\tblack{achieves lower relative time-limited $h^2$ errors than the tested alternatives} 
and is effective in situations where the unregularized method may deteriorate under noise.
\end{abstract}


\IEEEpeerreviewmaketitle

\section{Introduction} \label{sec:intro}
\IEEEPARstart{I}{n} engineering systems, the behavior over a finite horizon can govern system performance. In such cases, the input--output behavior over the finite horizon is evaluated using time-limited Gramians and finite horizon norms~\cite{gawronski1990model}.
In particular, when dealing with large-scale and complex systems, model order reduction (MOR) techniques are required to construct reduced-order models (ROMs)~\cite{antoulas2005approximation, gugercin2008h_2}.

In MOR for systems over a finite horizon, performance measures are defined on the finite horizon, and the goal is to construct a reduced-order system that preserves the input--output behavior over that horizon.
Well-known model-based approaches (i.e., system matrices available) for such systems include time-limited balanced truncation \cite{kurschner2018balanced} and time-limited $h^2$ (or $\mathcal{H}^2$) MOR \cite{goyal2019time, das2022h, sakamoto2025compression}.
On the other hand, in practice, one often encounters situations where the system matrices are unknown and only measured data are available, or where system identification~\cite{ljung1999systemID, pillonetto2010new} is difficult; in such cases, a data-driven framework is required, as shown in Figure~\ref{fig:our_approach}.

\tblack{Representative data-driven MOR methods for linear time-invariant (LTI) systems include the Loewner framework, which constructs interpolatory models from frequency-response data~\cite{mayo2007framework}, the eigensystem realization algorithm (ERA), which constructs a realization from impulse response data~\cite{juang1985eigensystem, H_Tu_2014}, and data-driven balanced truncation, which approximately extracts controllability and observability information from input--output data~\cite{gosea2022data_2}.
Furthermore, from the viewpoint of $h^2$-optimal approximation, realization-independent approaches~\cite{beattie2012realization}, time-domain data-driven approaches based on IRKA~\cite{ackermann2025time}, and recent methods that construct $h^2$-optimal ROMs by generating transfer-function information offline from time-domain or frequency-domain data have also been proposed~\cite{zulfiqar2024data}. 
In particular, when only impulse response data are available, ERA~\cite{juang1985eigensystem, H_Tu_2014} is a standard choice, but it does not generally guarantee an optimal ROM in the sense of the time-limited $h^2$ norm and may overfit to data noise.
}

\begin{figure}[htbp]
     \centering
     \includegraphics[width=0.8\linewidth]{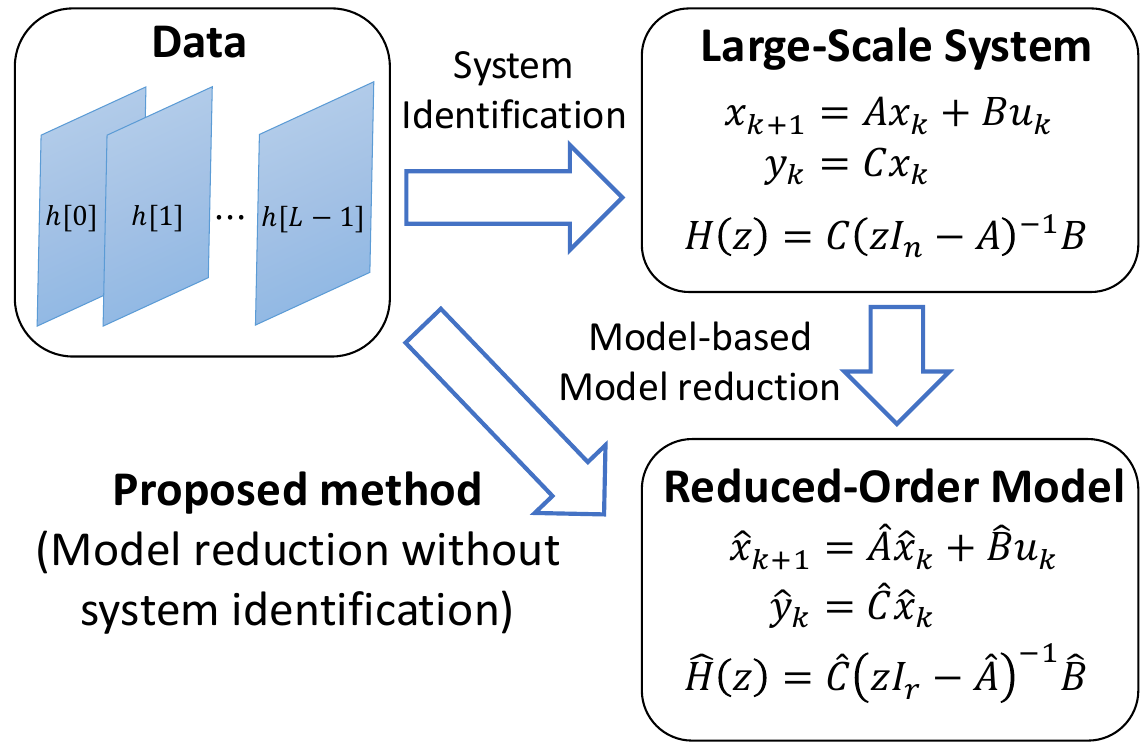}
     \caption{
     Overview of the proposed data-driven time-limited $h^2$ model reduction framework.\label{fig:our_approach}}
\end{figure}

Motivated by this limitation, in this paper, we consider a framework that directly solves the time-limited $h^2$ MOR problem using only noisy impulse response data, as shown in Figure~\ref{fig:our_approach}.
Unlike realization-independent or transfer-function-based data-driven $h^2$ MOR methods~\cite{beattie2012realization, ackermann2025time, zulfiqar2024data}, the proposed method directly optimizes a regularized time-limited $h^2$ objective using only noisy impulse response data.
In particular, the regularization introduced in this paper is inspired by kernel-based system identification, where kernels are used to encode prior information such as smoothness and stability of impulse responses~\cite{pillonetto2010new}. 
Accordingly, the contributions of this paper are summarized as follows:
\begin{itemize}
    \item \tblack{We formulate a regularized time-limited $h^2$ MOR problem for noisy impulse response data.}
    \item \tblack{We show that the resulting objective function and its gradient can be represented using only impulse response data, without using the system matrices of the original high-dimensional system.}
    \item \tblack{Through numerical experiments, we demonstrate the effectiveness of the proposed regularized method in situations where the performance of the unregularized method may deteriorate under noise.}
\end{itemize}

The remainder of this paper is organized as follows.
Section~\ref{sec:background} reviews the model-based time-limited $h^2$ model reduction method using the original system matrices.
Section~\ref{sec:formulation} formulates the data-driven time-limited $h^2$ model reduction problem using only noisy impulse response data.
Section~\ref{sec:method} shows that the gradient of the data-driven time-limited $h^2$ model reduction problem can be represented using the noisy impulse response data, and Section~\ref{sec:algorithm} introduces the corresponding optimization algorithm.
Section~\ref{sec:experiments} validates the proposed method using SLICOT benchmarks.
Section~\ref{sec:conclusion} concludes the paper.

\paragraph*{Notation}
For a vector $a$, $\|a\|$ denotes the Euclidean norm.
For a finite horizon $L$, define 
$\|x\|_{\ell^{2}_{L}}:=\big(\sum_{k=0}^{L-1}\|x_k\|^2\big)^{1/2}$ and 
$\|x\|_{\ell^{\infty}_{L}}:=\max_{0\le k\le L-1}\|x_k\|$.

For a matrix $A$, we denote the Frobenius norm, transpose, trace, the spectral radius, vectorization, and $(i,j)$th entry
by $\|A\|$, $A^{\top}$, $\tr A$, $\rho(A)$, $\operatorname{vec}(A)$, and $(A)_{i,j}$, respectively.

For a real-valued function $f$, $\mathrm{d}f$ denotes its (total) differential; for a perturbation $\mathrm{d}\theta$ of $\theta$, we write $\mathrm{d}f=\langle\nabla_{\theta}f,\mathrm{d}\theta\rangle$, where the inner product is $\langle A,B\rangle:=\mathrm{tr}(A^{\top}B)$.


\section{Model-based Time-Limited \texorpdfstring{$h^2$}{h2} MOR for LTI systems}\label{sec:background}

\subsection{Discrete-Time LTI systems over a finite horizon}\label{subsec:tlLTI}
Consider the discrete-time LTI system
\begin{empheq}[left=\empheqlbrace]{equation}
    \begin{aligned}
    x_{k+1} &= Ax_{k} + Bu_{k}, \\
    y_{k} &= Cx_{k},
    \end{aligned}
\label{eq:dt_lti}
\end{empheq}
where $x_k\in\mathbb{R}^{n}$, $u_k\in\mathbb{R}^{m}$, and $y_k\in\mathbb{R}^{p}$ denote the state, input, and output at time $k\in\mathbb{Z}_{\ge 0}$, respectively, and the matrices $A\in\mathbb{R}^{n\times n}$, $B\in\mathbb{R}^{n\times m}$, and $C\in\mathbb{R}^{p\times n}$ are constant system matrices. 
The impulse response (Markov parameters) and the transfer function of system~\eqref{eq:dt_lti} are given by
\begin{align}
    h[k] &:= C A^{k} B \in \mathbb{R}^{p\times m}, \label{eq:impulse_def}\\
    H(z) &:= C(zI-A)^{-1}B\in \C^{p\times m}. \label{eq:tf_H}    
\end{align}
Using the impulse response sequence $\{h[k]\}$ and assuming $x_0=0$, the input--output relation of system~\eqref{eq:dt_lti} over a finite horizon can be written as the convolution $y_{k+1} = \sum_{j=0}^{k} h[j]\: u_{k-j}$.

Under the assumption that system~\eqref{eq:dt_lti} is asymptotically stable, i.e., $A\in\R^{n\times n}$ is Schur stable ($\rho(A) < 1$), the $h^2$ norm of system~\eqref{eq:dt_lti} is defined as $\|H\|_{h^2}^2:= \sum_{k=0}^{\infty} \|h[k]\|^2$.
Following~\cite{goyal2019time,sakamoto2025compression}, the time-limited $h^2$ norm over the finite horizon $\{0,1,\dots,L-1\}$ is defined as
\begin{equation}
    \|H\|_{h^2_L}^2
    := \sum_{k=0}^{L-1} \|h[k]\|^2.
    \label{eq:tlH2_norm}
\end{equation}
Note that \eqref{eq:tlH2_norm} is well-defined even if $A$ is not Schur stable.

Next, we consider a ROM of system~\eqref{eq:dt_lti}:
\begin{empheq}[left=\empheqlbrace]{equation}
    \begin{aligned}
    \hat{x}_{k+1} &= \hat{A}\hat{x}_{k} + \hat{B}u_{k}, \\
    \hat{y}_{k} &= \hat{C}\hat{x}_{k},
    \end{aligned}
\label{eq:dt_r_lti}
\end{empheq}
with transfer function $\hat{H}(z) := \hat{C}(zI-\hat{A})^{-1}\hat{B}$,
where $\hat{x}_{k}\in \R^r$, $\hat{y}_{k}\in \R^{p}$, $\hat A\in\mathbb{R}^{r\times r}$, $\hat B\in\mathbb{R}^{r\times m}$, $\hat C\in\mathbb{R}^{p\times r}$, and $r\ll n$.
For the two systems~\eqref{eq:dt_lti} and \eqref{eq:dt_r_lti}, the following input--output error bound holds~\cite{goyal2019time}:
\begin{align}
    \|y-\hat y\|_{\ell_{L}^{\infty}}
    \le \|H-\hat H\|_{h^{2}_{L}}\cdot\|u\|_{\ell_{L}^{2}}.
    \label{eq:dt_lti_relation_finite}
\end{align}
Inequality~\eqref{eq:dt_lti_relation_finite} implies that, for inputs with bounded energy over the horizon, a sufficiently small time-limited $\|H-\hat H\|_{h^{2}_{L}}$ error guarantees that the output error $\|y-\hat y\|_{\ell^\infty_{L}}$ becomes correspondingly small.

\subsection{Time-Limited \texorpdfstring{$h^2$}{h2} MOR}\label{subsec:tlh2MOR}
Motivated by the input--output inequality~\eqref{eq:dt_lti_relation_finite}, we formulate the time-limited $h^2$ model reduction problem as follows:
\begin{equation}\label{prob:dt_lti_tlh2}
    \min_{\hat \theta\in \mathcal E}\ \|H-\hat{H}\|_{h^{2}_{L}}^2=\sum_{k=0}^{L-1} \|CA^k B\|^2 + f_{L}(\hat{\theta}),
\end{equation}
where $\hat\theta:=(\hat A, \hat B, \hat C)$ and $\mathcal E:=\R^{r\times r}\times \R^{r\times m}\times \R^{p\times r}$. Furthermore,
\begin{align}
    f_{L}(\hat\theta)
    &:=\tr (\hat C P_L \hat C^{\top}) - 2\tr (CR_L \hat C^{\top}) \\
    &=\tr(\hat B^{\top} Q_L \hat{B}) + 2\tr (B^{\top}S_{L}\hat{B}),
\end{align}
and 
\begin{align}
    P_L&:= \sum_{k=0}^{L-1}\hat{A}^{k}\hat{B}\hat{B}^{\top}(\hat{A}^{\top})^k, \label{eq:dt_tlLyap_PL}\\
    Q_L&:= \sum_{k=0}^{L-1}(\hat{A}^{\top})^k \hat{C}^{\top}\hat{C}\hat{A}^{k} \label{eq:dt_tlLyap_QL}\\
    R_L&:= \sum_{k=0}^{L-1}A^{k}B\hat{B}^{\top}(\hat{A}^{\top})^k, \label{eq:dt_tlSylve_RL}\\
    S_L&:=- \sum_{k=0}^{L-1}(A^{\top})^k C^{\top}\hat{C}\hat{A}^{k}. \label{eq:dt_tlSylve_SL}
\end{align}
Therefore, Problem~\eqref{prob:dt_lti_tlh2} can be equivalently rewritten as
\begin{equation}\label{prob:dt_lti_tlh2_mod}
    \min_{\hat{\theta}\in \mathcal E}\ f_{L}(\hat\theta).
\end{equation}

To solve Problem~\eqref{prob:dt_lti_tlh2_mod}, we recall the gradient formulas. 
For the objective function $f_{L}(\hat \theta)$, we denote the gradients with respect to $\hat A$, $\hat B$, and $\hat C$ as $\nabla_{\hat A}f_L$, $\nabla_{\hat B}f_L$, and $\nabla_{\hat C}f_L$, respectively.
\begin{prop}\label{prop:tl_gradients}
    Define 
    \begin{align}
        M &:= \sum_{j=0}^{L-1}(\hat A^\top)^{j}\,(R^\top (A^\top)^L C^\top \hat{C} \\
        &\quad\qquad- P(\hat{A}^{\top})^L \hat{C}^{\top}\hat{C})\,(\hat A^\top)^{L-1-j}.
    \end{align}
    Then, the gradients of $f_{L}(\hat\theta)$ are
    \begin{align}
        \nabla_{\hat A}f_L &= 2(Q_L\,\hat{A}\,P + S_{L}^\top A R + M) \\
        \nabla_{\hat B}f_L &= 2(Q_L\hat B + S_{L}^{\top}B) \\
        \nabla_{\hat C}f_L &= 2(\hat C P_L - CR_L), \\
    \end{align}
    where $P_L$, $Q_L$, $R_L$, and $S_L$ are defined in~\eqref{eq:dt_tlLyap_PL}, \eqref{eq:dt_tlLyap_QL}, \eqref{eq:dt_tlSylve_RL}, and \eqref{eq:dt_tlSylve_SL}.
    $P$ and $R$ are obtained by solving the following discrete-time Lyapunov equation and Sylvester equation
    \begin{align}
        \hat{A}P\hat{A}^{\top} + \hat{B}\hat{B}^{\top} &= P, \label{eq:inf_lyap_P}\\
        AR\hat{A}^{\top} + B\hat{B}^{\top} &= R. \label{eq:inf_sylve_R}
    \end{align}
\end{prop}
\begin{proof}
    The proof follows from~\cite{das2022h}.
\end{proof}

\section{Problem Formulation}\label{sec:formulation}
The goal of this paper is to solve Problem~\eqref{prob:dt_lti_tlh2_mod} using only noisy impulse response data of the system~\eqref{eq:dt_lti}.
Specifically, we consider noisy impulse response data of the form
\begin{align}
    \tilde{h}[k] = h[k] + \eta[k], \qquad k = 0, \ldots, L-1, \label{eq:noisy_data}
\end{align}
where $\eta[k]$ denotes noise. In particular, we assume that impulse response data of length $L$ are given as
\begin{align}
    \mathcal{D}_L
    := \big\{\, \tilde h[0],\,\tilde h[1],\,\dots,\,\tilde h[L-1] \,\big\}.
    \label{eq:impulse_data}
\end{align}

Using the impulse response data~\eqref{eq:impulse_data}, we reformulate Problem~\eqref{prob:dt_lti_tlh2_mod}.
First, multiplying $R_L$ and $S_L$ from the left by $C$ and $B^{\top}$, respectively, yields
\begin{align}
    & C R_L
    = \sum_{k=0}^{L-1} h[k] \hat{B}^{\top} (\hat{A}^{\top})^k
    =: Z_1,
    \label{eq:impulse_Z1}\\
    &B^{\top} S_L
    = -\sum_{k=0}^{L-1} h[k]^{\top} \hat{C} \hat{A}^k
    =: Z_2.
    \label{eq:impulse_Z2}
\end{align}
Therefore, if the exact impulse response data are available, the objective function $f_L(\hat{\theta})$ can be expressed equivalently as
\begin{align}
    f_{L,\mathrm{data}}(\hat{\theta})
    &:= \tr(\hat{C} P_L \hat{C}^{\top}) - 2 \tr(Z_1 \hat{C}^{\top}) \nonumber\\
    &= \tr(\hat{B}^{\top} Q_L \hat{B}) + 2 \tr(Z_2 \hat{B}).
    \label{eq:obj_data_exact}
\end{align}

\tblack{
In our setting, only the noisy data $\mathcal{D}_L$ in~\eqref{eq:impulse_data} are available.
Accordingly, instead of $Z_1$ and $Z_2$, we define
\begin{align}
    \tilde{Z}_1
    &:= \sum_{k=0}^{L-1} \tilde{h}[k] \hat{B}^{\top} (\hat{A}^{\top})^k,
    \label{eq:impulse_Z1_noisy} \\
    \tilde{Z}_2
    &:= -\sum_{k=0}^{L-1} \tilde{h}[k]^{\top} \hat{C} \hat{A}^k.
    \label{eq:impulse_Z2_noisy}
\end{align}
Using these quantities, we consider the data-based objective
\begin{align}
    \tilde{f}_{L,\mathrm{data}}(\hat{\theta})
    &:= \tr(\hat{C} P_L \hat{C}^{\top}) - 2 \tr(\tilde{Z}_1 \hat{C}^{\top}) \nonumber\\
    &= \tr(\hat{B}^{\top} Q_L \hat{B}) + 2 \tr(\tilde{Z}_2 \hat{B}).
    \label{eq:obj_data_noisy}
\end{align}
Note that $\tilde{f}_{L,\mathrm{data}}(\hat{\theta})$ coincides with $f_{L,\mathrm{data}}(\hat{\theta})$ in the noise-free case (i.e., $\eta[k]=0$).
}

\tblack{
If we minimize only $\tilde{f}_{L,\mathrm{data}}(\hat{\theta})$, the resulting ROM may overfit the noise in $\mathcal{D}_L$.
To suppress such overfitting and to incorporate prior information on desirable ROMs, we introduce a regularization term $\mathcal{R}(\hat{\theta})$ and consider the regularized optimization problem
\begin{align}
    \min_{\hat{\theta} \in \mathcal{E}}
    \ J_{L,\lambda}(\hat{\theta})
    := \tilde{f}_{L,\mathrm{data}}(\hat{\theta}) + \lambda \mathcal{R}(\hat{\theta}),
    \label{prob:dt_lti_tlh2_mod_reg}
\end{align}
where $\lambda \geq 0$ is a regularization parameter. 
Here, the regularization term is defined by
\begin{align}
    \mathcal{R}(\hat{\theta})
    := \tr\!\left( \hat{G}^{\top} K^{-1} \hat{G} \right),
\end{align}
where $K \succ 0$ and
\begin{align}
    \hat{G}:=
    \begin{bmatrix}
        \operatorname{vec}(\hat{h}[0])^{\top} \\
        \vdots \\
        \operatorname{vec}(\hat{h}[L-1])^{\top}
    \end{bmatrix}
    \in \mathbb{R}^{L \times pm},
\end{align}
with $\hat{h}[k] := \hat{C}\hat{A}^k\hat{B}$.
The term $\mathcal{R}(\hat{\theta})$ penalizes ROMs whose impulse responses are inconsistent with the prior information encoded by $K$, and thus helps construct a model that does not overfit the noisy data.
In particular, when the TC kernel~\cite{pillonetto2010new} is adopted, the entries of $K$ are given by
\begin{align}
    K_{ij} := \alpha_{K}^{\max(i,j)+1},
    \qquad i,j = 0,1,\ldots,L-1, \label{eq:tc_kernel}
\end{align}
where $0 < \alpha_{K} < 1$.
This kernel encodes the prior that the impulse response decays smoothly, which is suitable for constructing a ROM that reflects the expected structure of the true system while being robust to noise.
}

\tblack{
Thus, the problem addressed in this paper is stated as follows.
\begin{problem}\label{prob:impulse_data_driven_h2}
    Develop an algorithm for solving the regularized time-limited $h^2$ MOR problem~\eqref{prob:dt_lti_tlh2_mod_reg} in a data-driven manner.
    Specifically, using the noisy impulse response data $\mathcal{D}_L$, construct a ROM that reflects prior information and avoids overfitting to noise, under the condition that the matrices $A$, $B$, and $C$ in~\eqref{eq:dt_lti} are unknown.
\end{problem}
Note that Proposition~\ref{prop:tl_gradients} provides gradient formulas for the original time-limited $h^2$ objective.
However, when $A$, $B$, and $C$ are unknown, these gradients cannot be computed directly.
Therefore, it is necessary to derive gradient formulas that depend only on the available noisy impulse response data.
\begin{remark}\label{rem:comparison_sysID}
    As an alternative to the proposed method, one may consider combining kernel-based system identification~\cite{pillonetto2010new} with ERA~\cite{juang1985eigensystem,H_Tu_2014} (SysID+ERA). 
    More precisely, from the noisy impulse response data $\mathcal{D}_{L}$, we first construct input--output data, then estimate an impulse response by the kernel-based system identification~\cite{pillonetto2010new}, and finally construct an $r$th-order realization by ERA.
    Thus, SysID+ERA involves two sources of error, namely the impulse response estimation error and the realization error introduced by ERA.
    Even if the identification step is successful, ERA applied to the accurately estimated data does not necessarily yield a ROM that is optimal in the time-limited $h^2$ sense. 
    If the identification step fails, then the subsequent ERA step cannot be expected to construct a desired ROM reflecting prior information, due to overfitting to noise. 
    By contrast, with an appropriate regularization term, 
    \tblack{the proposed method can directly optimize a regularized time-limited $h^2$ objective} 
    even under such circumstances. 
    These points will be examined numerically in Section~\ref{sec:experiments}.
\end{remark}
}

\section{Data-Driven Time-Limited \texorpdfstring{$h^2$}{h2} MOR from Impulse Response Data}

\subsection{Gradients from Impulse Response Data}\label{sec:method}
\tblack{
We reconstruct the gradient of $J_{L,\lambda}$ using the noisy impulse response data.
Here, we denote the gradient of $J_{L,\lambda}$ by
\begin{align}
    \nabla J_{L,\lambda}
    =
    \big(
    \nabla_{\hat A}J_{L,\lambda},
    \nabla_{\hat B}J_{L,\lambda},
    \nabla_{\hat C}J_{L,\lambda}
    \big).
\end{align}
\begin{theorem}\label{thm:tl_gradients}
    Define 
    \begin{align}
    \hat{\xi}[k]:=
        \sum_{j=0}^{L-1}
        \big(K^{-1}\big)_{k+1,j+1}\hat{h}[j].
    \end{align}
    Then, the gradient $\nabla J_{L,\lambda}$ of $J_{L,\lambda}$ is given by
    \begin{align}
        \nabla_{\hat A}J_{L,\lambda}
        &=
        2
        \sum_{k=1}^{L-1}
        \sum_{i=0}^{k-1}
        (\hat A^{k-1-i})^{\top}
        \hat C^{\top}
        \big(
        \hat C\hat A^k\hat B
        -
        \tilde h[k] \\
        &\qquad\qquad\quad+
        \lambda \hat{\xi}[k]
        \big)
        \hat B^{\top}
        (\hat A^i)^{\top},\label{eq:J_grad_Ahat}
        \\
        \nabla_{\hat B}J_{L,\lambda}
        &=
        2\Bigg(Q_L\hat B+\tilde Z_2^{\top}+\lambda\sum_{k=0}^{L-1}(\hat A^k)^{\top}\hat C^{\top}\hat{\xi}[k]
        \Bigg), \label{eq:J_grad_Bhat}\\
        \nabla_{\hat C}J_{L,\lambda}
        &=2\Bigg(\hat C P_L-\tilde Z_1+\lambda\sum_{k=0}^{L-1}\hat{\xi}[k]\hat B^{\top}(\hat A^k)^{\top}\Bigg), \label{eq:J_grad_Chat}
    \end{align}
    where $P_L$, $Q_L$, $\tilde Z_1$, and $\tilde Z_2$ are defined in
    \eqref{eq:dt_tlLyap_PL},
    \eqref{eq:dt_tlLyap_QL},
    \eqref{eq:impulse_Z1_noisy}, and
    \eqref{eq:impulse_Z2_noisy}, respectively.
\end{theorem}
\begin{proof}
    See Appendix.
\end{proof}
}

The gradient characterized by Theorem~\ref{thm:tl_gradients} is a data-based representation of the gradient of the regularized objective $J_{L,\lambda}$.
In particular, in the noise-free case, if equation~\eqref{eq:impulse_def} holds, then the gradient shown in Theorem~\ref{thm:tl_gradients} coincides with the gradient obtained from Proposition~\ref{prop:tl_gradients} applied to the regularized objective.

\begin{remark}\label{rem:comparison_gradient}
In our setting, where $(A,B,C)$ are unavailable but $\mathcal D_L$ is available, using the gradient in Proposition~\ref{prop:tl_gradients} for the data-fitting term requires system identification~\cite{ljung1999systemID}.
From the viewpoints of the additional computational cost (typically of order $\mathcal{O}(n^3)$) and the identification error, it is preferable to use the gradient in Theorem~\ref{thm:tl_gradients}.
Furthermore, since evaluating the gradient in Theorem~\ref{thm:tl_gradients} is independent of $n$, it is also preferable for large-scale systems.
\end{remark}

\subsection{Impulse Response Data-Driven \texorpdfstring{$h^2$}{h2} MOR Algorithm}\label{sec:algorithm}
We present Algorithm~\ref{alg:ABCunknown}, which solves Problem~\eqref{prob:dt_lti_tlh2_mod_reg} using the noisy impulse response data $\mathcal D_L$.
Because Problem~\eqref{prob:dt_lti_tlh2_mod_reg} is unconstrained on $\mathcal E$, if the gradient in Theorem~\ref{thm:tl_gradients} is zero at a point in $\mathcal E$, that point is a stationary point.

We briefly describe Algorithm~\ref{alg:ABCunknown}.
First, we generate the initial reduced-order system matrices
$\hat{\theta}_{1}=(\hat{A}_{(1)}, \hat{B}_{(1)}, \hat{C}_{(1)})\in \mathcal E$
\tblack{using existing methods (see Remark~\ref{rem:how_to_choose_initial_point}).}
At each iteration $\ell$, after computing the gradient, the reduced matrices
$\hat{\theta}_{\ell}:=(\hat A_{(\ell)},\hat B_{(\ell)},\hat C_{(\ell)})$
are updated to satisfy the Armijo condition
\begin{align}
    J_{L,\lambda}(\bar\theta)
    \leq
    J_{L,\lambda}(\hat\theta_{\ell})
    -
    c_1 \alpha_{\ell}
    \|
    \nabla J_{L,\lambda}(\hat \theta_{\ell})
    \|^2
\end{align}
in the backtracking loop (Lines 5--10).

The convergence of Algorithm~\ref{alg:ABCunknown} is guaranteed by the following standard convergence result.
\begin{theorem}\label{thm:tlh2_convergence}
    Assume that the sequence $\{\hat\theta_\ell\}$ with
    $\hat\theta_\ell=(\hat A_{(\ell)},\hat B_{(\ell)},\hat C_{(\ell)})$
    generated by Algorithm~\ref{alg:ABCunknown} with $\mathrm{tol}=0$ is bounded.
    Then $\{\hat\theta_\ell\}$ converges to a stationary point of Problem~\eqref{prob:dt_lti_tlh2_mod_reg}.
\end{theorem}

\begin{proof}
    It follows from~\cite[Thm. 3.2]{attouch2013convergence}.
    See Appendix.
\end{proof}

\begin{remark}
When $(A,B,C)$ are known, one can construct the same algorithm as Algorithm~\ref{alg:ABCunknown} by using the gradient in Proposition~\ref{prop:tl_gradients} for the data-fitting term together with the gradient of the regularization term.
In particular, if equation~\eqref{eq:impulse_def} holds and
$\hat{\theta}_{1}$,
$\alpha_{\mathrm{init}}$,
$\beta$,
$c_1$,
$\lambda$, and
$\mathrm{tol}$
are fixed, then the sequence generated by Algorithm~\ref{alg:ABCunknown} coincides with the sequence generated by the corresponding model-based regularized $h^2$ MOR algorithm.

On the other hand, as stated in Remark~\ref{rem:comparison_gradient}, using Proposition~\ref{prop:tl_gradients} for the data-fitting term yields a gradient computation that depends on $n$.
In particular, solving the Sylvester equation~\eqref{eq:inf_sylve_R} using the Bartels--Stewart method~\cite{simoncini2016computational} costs $\mathcal O(n^3)$, and this term becomes dominant when $n$ is large.
Algorithm~\ref{alg:ABCunknown} is based on the gradient in Theorem~\ref{thm:tl_gradients} and does not depend on $n$.
\end{remark}

\begin{remark}\label{rem:how_to_choose_initial_point}
Since Problem~\eqref{prob:dt_lti_tlh2_mod_reg} is nonconvex, the choice of the initial point is important when applying Algorithm~\ref{alg:ABCunknown}. 
Methods such as ERA~\cite{juang1985eigensystem,H_Tu_2014} and data-driven balanced truncation~\cite{gosea2022data_2} can also construct ROMs for discrete-time LTI systems from impulse response data alone. 
As shown in Section~\ref{sec:experiments}, using a ROM obtained by such methods as an initial point 
\tblack{leads to a model with a lower time-limited $h^2$ objective value than using a randomly generated initial point.}
\end{remark}

\begin{figure}[!t]
\begin{algorithm}[H]
    \caption{Gradient Method for Problem~\eqref{prob:dt_lti_tlh2_mod_reg}}
    \label{alg:ABCunknown}
    \begin{algorithmic}[1]
    \REQUIRE Length-$L$ noisy impulse response data $\mathcal{D}_L$, kernel matrix $K$, regularization parameter $\lambda\geq 0$, initial point $\hat{\theta}_{1}\in \mathcal{E}$, initial step-size $\alpha_{\mathrm{init}}>0$, backtracking parameter $\beta\in(0,1)$, Armijo parameter $c_1>0$, tolerance $\mathrm{tol}>0$
    \ENSURE $\hat{\theta}\in \mathcal{E}$
    \FOR{$\ell=1,2,\ldots$}
        \STATE Compute $J_{\ell}:=J_{L,\lambda}(\hat{\theta}_{\ell})$ and $g_{\ell}:=\nabla J_{L,\lambda}(\hat{\theta}_{\ell})$ using \eqref{eq:obj_data_noisy}, Theorem~\ref{thm:tl_gradients}, and $\mathcal R(\hat{\theta}_{\ell})$
        \STATE \textbf{if} $\|g_{\ell}\| < \mathrm{tol}$ \textbf{then break}
        \STATE $\alpha_{\ell}=\alpha_{\mathrm{init}}$
        \WHILE{true}
            \STATE $\bar{\theta}=\hat{\theta}_{\ell}-\alpha_{\ell}g_{\ell}$
            \STATE \tblack{Compute $\bar{J}:=J_{L,\lambda}(\bar{\theta})$ using \eqref{eq:obj_data_noisy} and $\mathcal R(\bar{\theta})$}
            \STATE \textbf{if} $\bar{J}\leq J_{\ell}-c_1\alpha_{\ell}\|g_{\ell}\|^2$ \textbf{then} $\hat{\theta}_{\ell+1}=\bar{\theta}$; \textbf{break}
            \STATE $\alpha_{\ell}\leftarrow \beta\alpha_{\ell}$
        \ENDWHILE
    \ENDFOR
    \end{algorithmic}
\end{algorithm}
\end{figure}
\section{Experiments}\label{sec:experiments}
\tblack{
In this section, we evaluate the proposed method using the CD player model in the SLICOT benchmark collection~\cite{ChahlaouiVanDooren2002}.
In Algorithm~\ref{alg:ABCunknown}, we set $\alpha_{\mathrm{init}}=1$, $\beta=0.5$, $c_1=10^{-4}$, and $\mathrm{tol}=10^{-5}$.
In all experiments, we evaluate the obtained ROMs by the relative error
$\frac{\|H-\hat{H}\|_{h^2_L}}{\|H\|_{h^2_L}}$.
Since the CD player benchmark is given as a continuous-time LTI system, we discretize it by zero-order hold with sampling time $5\times 10^{-2}$ and obtain system~\eqref{eq:dt_lti} with $n=120$, $m=2$, and $p=2$.
We set the horizon length to $L=500$ and the reduced order to $r=10$.
To assess robustness to noise, we generate the noisy impulse response data $\mathcal D_{L}$ by adding Gaussian noise scaled by $h_{\mathrm{scale}} := \max_{0 \leq k \leq L-1} \|h[k]\|$.
Specifically, we define $\eta[k]:=\sigma h_{\mathrm{scale}} \epsilon[k]$ in \eqref{eq:noisy_data}, where the entries of $\epsilon[k]$ are independently drawn from the standard normal distribution.
As the initial point, we use the ROM generated by ERA.
}

\tblack{
We first consider the case of small noise ($\sigma=10^{-3}$), that is, a situation where system identification has been successful and a high-precision impulse response has been obtained.
Figure~\ref{fig:cd_smallnoise} shows the convergence behavior of Algorithm~\ref{alg:ABCunknown} with $\lambda=0$ for two initializations: ERA-based initialization and random initialization.
In this case, the relative error is improved from these initial points, and 
\tblack{the proposed method converges to ROMs with small relative time-limited $h^2$ errors.}
In particular, \tblack{using ERA as the initial point allowed us to construct a ROM with a lower relative time-limited \(h^2\) error.}
This result indicates that, when the noise level is sufficiently small, the proposed method can still work well even without regularization.
Note that, in the noisy case, the proposed algorithm decreases the surrogate objective $J_{L,\lambda}(\hat \theta)$ constructed from noisy impulse response data, rather than the true time-limited $h^2$ MOR objective $f_{L}(\hat\theta)$. Therefore, although the surrogate objective can be reduced along the iterations, monotonic decrease of the true time-limited $h^2$ error is not guaranteed in general.}
\tblack{This discrepancy can explain the oscillatory behavior of the relative time-limited $h^2$ error observed in Figures~\ref{fig:cd_smallnoise} and~\ref{fig:cd_largenoise_conv}.}

\begin{figure}[htbp]
    \centering
    \includegraphics[width=8.25cm]{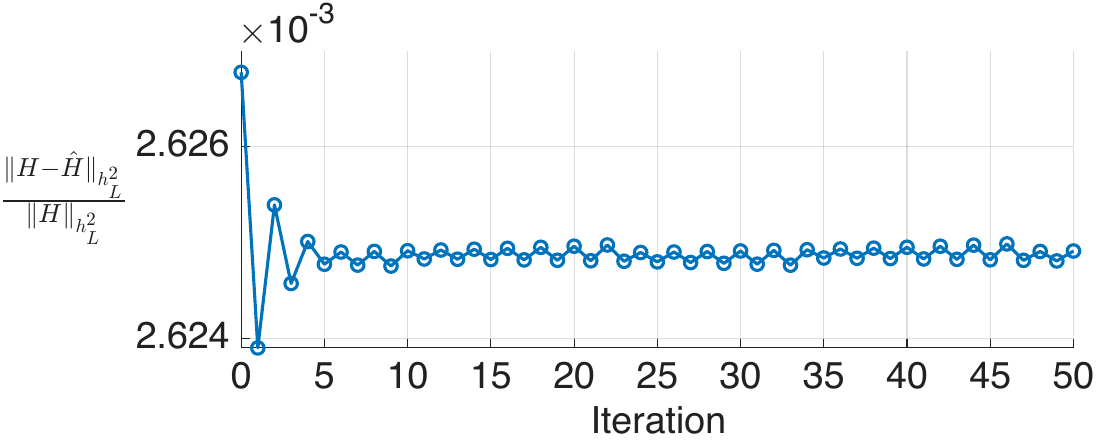}\\[2mm]
    \includegraphics[width=8.25cm]{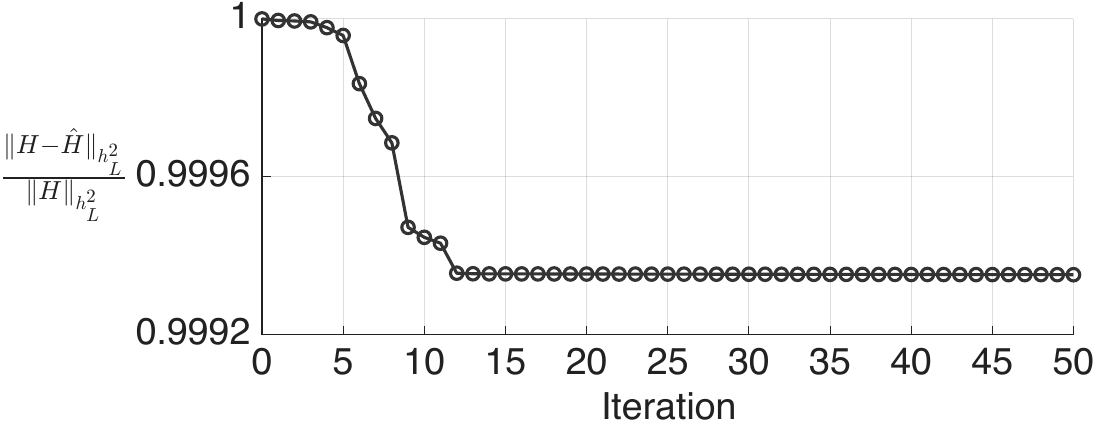}
    \caption{Convergence behavior of Algorithm~\ref{alg:ABCunknown} using impulse response data corrupted by small noise ($\sigma=10^{-3}$). The upper plot shows the result with ERA-based initialization, and the lower plot shows the result with random initialization.}
    \label{fig:cd_smallnoise}
\end{figure}

We next consider the case of large noise ($\sigma=0.5$) and investigate the effect of regularization. Note that we use the TC kernel~\eqref{eq:tc_kernel} with $\alpha_K=0.99$.
Figure~\ref{fig:cd_largenoise_conv} shows the convergence behavior of the proposed method for $\lambda=0$, $\lambda=0.0005$, $\lambda=0.001$, and $\lambda=0.002$.
In the unregularized case $\lambda=0$, the algorithm converges to a solution whose relative error is worse than that of the ERA-based initialization. This is presumably because the algorithm overfits the surrogate objective constructed from noisy impulse response data, which does not necessarily lead to a decrease in the true time-limited $h^2$ error.
In contrast, when $\lambda>0$, the proposed method converges to solutions whose relative errors are improved from the initial point.
These results indicate that regularization is effective in suppressing overfitting to noisy impulse response data and in 
\tblack{guiding the optimization toward ROMs with lower relative time-limited $h^2$ errors.}

\begin{figure}[htbp]
    \centering
    \includegraphics[width=8.25cm]{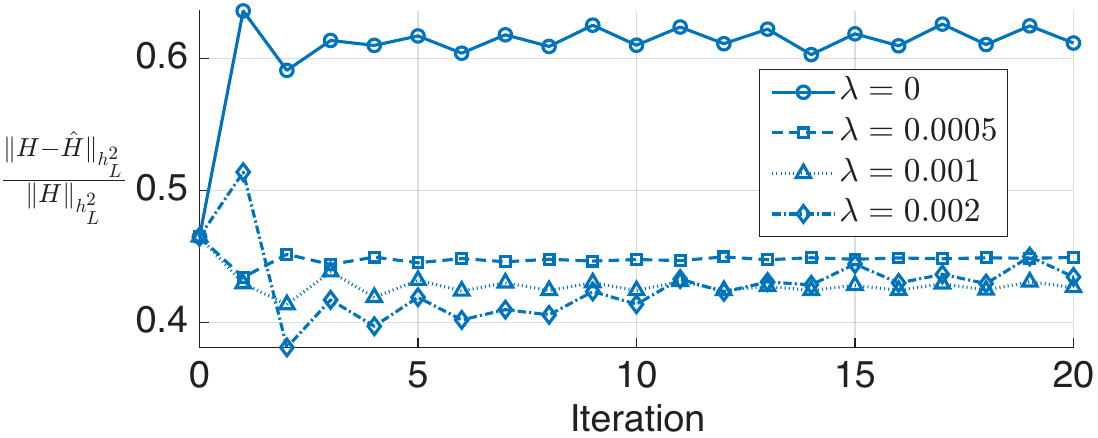}

    \vspace{2mm}

    \includegraphics[width=8.25cm]{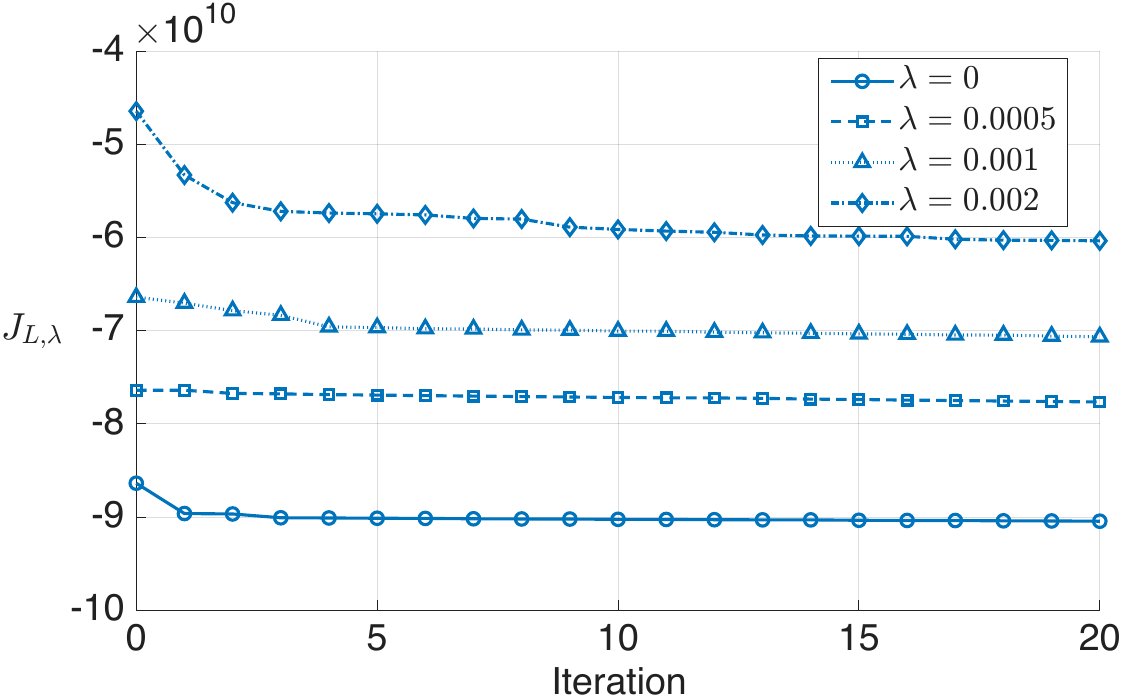}
    \caption{Convergence behavior of Algorithm~\ref{alg:ABCunknown} for the CD player model using impulse response data corrupted by large noise ($\sigma=0.5$). The upper plot shows the relative time-limited $h^2$ error, and the lower plot shows the objective value of the optimization problem.}
    \label{fig:cd_largenoise_conv}
\end{figure}

\tblack{For comparison, we also applied the method SysID+ERA, which is described in Remark~\ref{rem:comparison_sysID}.
As shown in Table~\ref{tab:cd_largenoise_reg}, among the tested methods, the proposed regularized method attained the best performance.
Note that, for SysID+ERA and Algorithm~\ref{alg:ABCunknown}, the regularization parameter was selected from several candidates so as to achieve the best performance.
This result suggests that directly solving the regularized time-limited $h^2$ MOR problem from noisy impulse response data is more effective than first estimating the impulse response and then constructing a reduced-order realization.}
\tblack{However, for the extreme noise level $\sigma=50$, the relative error remains large even with the proposed regularization.
This indicates that the method can mitigate overfitting to noisy impulse responses, whereas its performance may deteriorate when the effective signal-to-noise ratio becomes too low.}

\begin{table}[htbp]
    \caption{Relative time-limited $h^2$ errors of SysID+ERA, ERA, and Alg.~\ref{alg:ABCunknown} for the CD player model under different noise levels.}
    \label{tab:cd_largenoise_reg}
    \centering
    \resizebox{0.95\linewidth}{!}{
    \begin{tabular}{c|c|c|c}
         & SysID+ERA & ERA & Alg.~\ref{alg:ABCunknown} with ERA \\
        \hline
        $\sigma=0.5$ & $0.457$ & $0.465$ & $\mathbf{0.427}$ \\
        \hline
        $\sigma=5$ & $4.63$ & $4.64$ & $\mathbf{1.90}$ \\
        \hline
        $\sigma=50$ & $45.10$ & $45.17$ & $\mathbf{22.42}$
    \end{tabular}
    }
\end{table}

\section{Conclusion}\label{sec:conclusion}
We developed a data-driven \tblack{regularized} time-limited $h^2$ model reduction method for discrete-time LTI systems using \tblack{noisy} impulse response data.
In particular, we showed that the objective function and its gradient can be represented using only impulse response data, without using the original system matrices.
Numerical experiments demonstrated that the proposed regularized method is effective in situations where the unregularized method may deteriorate under noise.

\section*{Acknowledgment}
This work was supported by JSPS KAKENHI under Grant Numbers 23K28369 and 25KJ0986.

\appendices
\section*{APPENDIX}
\makeatletter
\renewcommand{\appendixname}{}
\renewcommand{\thesectiondis}[2]{\Alph{section}.}
\makeatother

\begin{proof}[Proof of Theorem~\ref{thm:tl_gradients}]
    Let $\tilde E_k := \hat C \hat A^k \hat B - \tilde h[k]$. Then, $\tilde f_{L,\mathrm{data}}$ is written as
    $\tilde f_{L,\mathrm{data}}=\sum_{k=0}^{L-1}\tr(\tilde E_k^{\top}\tilde E_k)-\sum_{k=0}^{L-1}\tr(\tilde h[k]^{\top}\tilde h[k])$.
    Hence,
    \begin{align}
        \mathrm{d}\tilde f_{L,\mathrm{data}}=2\sum_{k=0}^{L-1}\tr\!\left(\tilde E_k^{\top}\,\mathrm{d}\hat h[k]\right).
    \end{align}
    Furthermore, since $\mathcal{R}(\hat\theta)=\tr(\hat G^{\top}K^{-1}\hat G) = \sum_{k=0}^{L-1}\tr\!\left(\hat h[k]^{\top}\hat{\xi}[k]\right)$ and $K^{-1}$ is symmetric, we have 
    \begin{align}
        \mathrm{d}\mathcal{R}(\hat\theta)=2\sum_{k=0}^{L-1}\tr\!\left(\hat\xi[k]^{\top}\,\mathrm{d}\hat h[k]\right).
    \end{align}
    
    For $\hat A$, with $\hat B$ and $\hat C$ fixed, we have
    \begin{align}
        \mathrm{d}\hat h[k]
        &=
        \hat C\,\mathrm{d}(\hat A^k)\hat B =
        \hat C
        \left(
        \sum_{i=0}^{k-1}
        \hat A^{k-1-i}(\mathrm{d}\hat A)\hat A^i
        \right)\hat B.
    \end{align}
    Hence,
    \begin{align}
        \mathrm{d}J_{L,\lambda}
        &=\mathrm{d}\tilde f_{L,\mathrm{data}}+\lambda \mathrm{d}\mathcal{R}(\hat\theta) \\
        &=2\sum_{k=1}^{L-1}\sum_{i=0}^{k-1}
        \tr\!\left(
        \big(
        \tilde E_k+\lambda \hat\xi[k]
        \big)^{\top}
        \hat C\hat A^{k-1-i}
        (\mathrm{d}\hat A)\hat A^i\hat B
        \right).
    \end{align}
    Thus, by the definition of the gradient $\mathrm{d}J_{L,\lambda}=\tr\!\left((\nabla_{\hat A}J_{L,\lambda})^{\top}\mathrm{d}\hat A\right)$ and the properties of traces, \eqref{eq:J_grad_Ahat} holds.

    Differentiating $\tilde f_{L,\mathrm{data}}$ in~\eqref{eq:obj_data_noisy} and $\mathcal{R}(\hat\theta)$ with respect to $\hat B$ and $\hat C$ yields equations~\eqref{eq:J_grad_Bhat} and~\eqref{eq:J_grad_Chat}.
\end{proof}

\begin{proof}[Proof of Theorem~\ref{thm:tlh2_convergence}]
    We define the closure
    $\mathcal K:=\overline{\{\hat{\theta}_\ell\mid\ \ell\ge1\}}$ of the sequence generated by Algorithm~\ref{alg:ABCunknown}.
    Since $\{\hat{\theta}_{\ell}\}$ is bounded, $\mathcal K$ is a compact subset of ${\mathcal E}$.
    Furthermore, since $J_{L,\lambda}$ is smooth, $\nabla^2 J_{L,\lambda}$ is continuous on $\mathcal K$; hence, by the extreme value theorem,
    $L_{\nabla}:=\sup_{x\in \mathcal K}\|\nabla^2 J_{L,\lambda}(x)\|_{\mathrm{op}}<\infty$, where $\|\cdot \|_{\mathrm{op}}$ denotes the induced operator norm.
    Therefore, for all $x,y\in \mathcal K$,
    \begin{align}\label{eq:lip}
        \|\nabla J_{L,\lambda}(x)-\nabla J_{L,\lambda}(y)\|
        \le L_{\nabla}\|x-y\|.
    \end{align}

    Next, letting $s_{\ell}:=\hat{\theta}_{\ell+1}-\hat{\theta}_{\ell}
        =-\alpha_{\ell}\nabla J_{L,\lambda}(\hat{\theta}_{\ell})$,
    we obtain from the Armijo condition in Algorithm~\ref{alg:ABCunknown} that
    \begin{align}\label{eq:ABS_H1}
        J_{L,\lambda}(\hat\theta_{\ell + 1})
        &\leq J_{L,\lambda}(\hat{\theta}_{\ell})
        -
        \frac{c_1}{\alpha_{\mathrm{init}}}\|s_{\ell}\|^{2}.
    \end{align}

    By the Descent lemma~\cite[Lem.~3.1]{attouch2013convergence} and~\eqref{eq:lip}, the Armijo condition is satisfied whenever $\alpha \le \frac{2(1-c_1)}{L_{\nabla}}$.
    Hence, the backtracking procedure terminates, and the accepted step size satisfies $\alpha_{\ell}\geq \underline{\alpha}:=\beta\cdot\min \left\{\alpha_{\mathrm{init}}, \frac{2(1-c_1)}{L_{\nabla}}\right\}$.
    As a result, we obtain
    \begin{align}
        \|\nabla J_{L,\lambda}(\hat{\theta}_{\ell + 1})\|
        &\leq \|\nabla J_{L,\lambda}(\hat{\theta}_{\ell})\| + L_{\nabla}\|s_\ell\| \\
        &\leq (\underline{\alpha}^{-1}+ L_{\nabla})\|s_\ell\|,\label{eq:ABS_H2}
    \end{align}
    where the first inequality follows from~\eqref{eq:lip}, and the second from the definition of $s_\ell$ and the lower bound of $\alpha_\ell$.

    Lastly, $\{\hat{\theta}_{\ell}\}$ is bounded and, by the Bolzano--Weierstrass theorem, there exists a convergent subsequence $\{\hat\theta_{\ell_{j}}\}$ whose limit lies in $\mathcal K$.
    Furthermore, since $J_{L,\lambda}$ is real-analytic, it is a KL function.
    Hence, by \eqref{eq:ABS_H1}, \eqref{eq:ABS_H2}, and~\cite[Thm.~3.2]{attouch2013convergence}, Algorithm~\ref{alg:ABCunknown} converges to a stationary point in $\mathcal K$.
\end{proof}


\bibliographystyle{IEEEtran}
\bibliography{main.bib}

\end{document}